\documentclass[11pt]{article}
\usepackage{graphicx}

\usepackage{amsmath,amssymb,amsbsy,amsfonts,amsthm,latexsym,
            amsopn,amstext,amsxtra,euscript,amscd,color,breqn,mathrsfs}

\PassOptionsToPackage{hyphens}{url}\usepackage{hyperref}
    
\usepackage[capbesideposition=outside,capbesidesep=quad]{floatrow}

\restylefloat{table}
            
\usepackage{multirow,caption}

\newfont{\teneufm}{eufm10}
\newfont{\seveneufm}{eufm7}
\newfont{\fiveeufm}{eufm5}

\usepackage{systeme}

\usepackage{longtable}

\usepackage{xcolor}
\usepackage{tikz}
\usetikzlibrary{shapes}
\usetikzlibrary{decorations.markings}

\usepackage{pstricks}
\usetikzlibrary{shapes,arrows}
 \usepackage{pstricks-add}
\usepackage{epsfig}
 \usepackage[space]{grffile} 
 \usepackage{etoolbox} 
 \makeatletter 
 \patchcmd\Gread@eps{\@inputcheck#1 }{\@inputcheck"#1"\relax}{}{}
 \makeatother

\newtheorem{thm}{Theorem}

\newtheorem{cor}[thm]{Corollary}

\newtheorem{rem}[thm]{Remark}

\newtheorem{exa}[thm]{Example}

\newcommand{\Tr}{{\rm Tr}}
\newcommand{\Trn}{{\rm Tr}_n}

\newcommand{\cB}{\mathscr{B}}

\def\+{\oplus}

\def\cB{{\mathcal B}}

\def\cW{{\mathcal W}}

\def\F{{\mathbb F}}

\def\N{{\mathbb N}}

\def\uu{{\bf u}}

\def\00{{\bf 0}}
\def\11{{\bf 1}}
\def\+{\oplus}

\def\\{\cr}
\def\({\left(}
\def\){\right)}

\newcommand{\cardinality}[1]{\# #1}

\def\cW{{\mathcal W}}

\newcommand{\g}{}

\usepackage[colorinlistoftodos,prependcaption,textsize=tiny]{todonotes}

\providecommand{\newoperator}[3]{%
  \newcommand*{#1}{\mathop{#2}#3}}

\newoperator{\FD}{\mathrm{FD}}{\nolimits}

\begin{document}
 
 \title{\bf Low $c$-differentially uniform functions via an extension of Dillon's switching method}
\author{\Large Chunlei Li$^1$, Constanza Riera$^2$,  Pantelimon St\u anic\u a$^3$\footnote{Corresponding author}\\
\vspace{.01cm}\\
\small $^1$ Department of Informatics, University of Bergen,\\
\small 5020, Bergen, Norway;
\small  {\tt chunlei.li@uib.no}\\
\small $~^2$Department of Computer Science, Electrical Engineering \\
\small and Mathematical Sciences,\\
\small  Western Norway University of Applied Sciences, \\
\small 5020 Bergen, Norway; {\tt csr@hvl.no}\\ 
\small $^3$ Department of Applied Mathematics, Naval Postgraduate School\\
\small Monterey, CA 93943-5212, U.S.A.; {\tt pstanica@nps.edu}
}
\date{}

\maketitle

\begin{abstract}
In this paper we generalize Dillon's switching method to characterize the exact $c$-differential uniformity of functions constructed via this method. 
More precisely, we modify some PcN/APcN and other functions with known $c$-differential uniformity in a controllable number of coordinates to render more such functions. We present several applications of the method in constructing PcN and APcN functions with respect to all $c\neq 1$. As a byproduct, we generalize some result of [Y. Wu, N. Li, X. Zeng,
{\em New PcN and APcN functions over finite fields}, 
Designs Codes Crypt. 89 (2021), 2637--2651]. Computational  results rendering functions with low differential uniformity, as well as, other good cryptographic properties are sprinkled throughout the paper.
\end{abstract}
 
 \section{Background}
 As customary, for a positive integer $n$, we let $\F_{p^n}$ denote the  finite field with $p^n$ elements, and $\F_{p^n}^*=\F_{p^n}\setminus\{0\}$ (for $a\neq 0$, by $\frac{1}{a}$ we mean the inverse of $a$). Further, let $\F_p^m$ denote the $m$-dimensional vector space over $\F_p$.
 The cardinality of a set $S$ is denoted by $\cardinality{S}$.
We call a function from $\F_{p^n}$ to $\F_p$  a {\em Boolean} (for $p=2$) or {\em $p$-ary} (for $p>2$) {\em function} on $n$ variables. 
For $m\,|\,n$, we let  the {\em relative trace} be defined by $\Tr_{p^n/p^m}(x)=\Tr_m^n(x)=\sum_{i=0}^{n/m-1} x^{p^{mi}}$.  When $m=1$, we will denote this {\em absolute trace} by $\Tr_n$ (abusing notation, for $q=p^t$, where $t>0,n\geq2$ are integers, we will denote by $\Trn$ the absolute trace of $\F_{q^n}$ over $\F_q$, when the base field $\F_q$ is clear from the context).
 Given a $p$-ary function $f$, its derivative with respect to~$a \in \F_{p^n}$ is the function
$
 D_{a}f(x) =  f(x + a)- f(x), \mbox{ for  all }  x \in \F_{p^n}.
$
For positive integers $n$ and $m$, any map $F:\F_{p^n}\to\F_{p^m}$ is called a {\em vectorial $p$-ary  function}, or {\em  $(n,m)$-function}. 
When $m=n$, $F$ can be uniquely represented as a univariate polynomial over $\F_{p^n}$  of the form
$
F(x)=\sum_{i=0}^{p^n-1} a_i x^i,\ a_i\in\F_{p^n}.
$
Given an $(n,n)$-function $F$, and $a,b\in\F_{p^n}$, we let $\Delta_F(a,b)=\cardinality{\{x\in\F_{p^n} : F(x+a)-F(x)=b\}}$. Then
$\Delta_F=\max\{\Delta_F(a,b)\,:\, a,b\in \F_{p^n}, a\neq 0 \}$ is the {\em differential uniformity} of $F$. If $\Delta_F= \delta$, then we say that $F$ is {\em differentially $\delta$-uniform}. If $\delta=1$, then $F$ is called a {\em perfect nonlinear} ({\em PN}) function, or {\em planar} function. If $\delta=2$, then $F$ is called an {\em almost perfect nonlinear} ({\em APN}) function. It is well known that PN functions do not exist if $p=2$.

Wagner~\cite{DW99} introduced an extension of the differential attack method against block ciphers   known as  the boomerang attack.  Cid et al.~\cite{cid} proposed  a theoretical tool, the Boomerang Connectivity Table (BCT) to analyze the resistance of a block cipher against this attack, and furthermore, Boura and Canteaut~\cite{BC18}  proposed a   quantifier of  the resistance of a function against the boomerang attack, namely, the boomerang uniformity, which is  the maximum value in the BCT excluding the first row and first column. Li et al.~\cite{KLi19} proposed an equivalent formulation that avoids using inverses in the original definition of the boomerang uniformity. For any $a,b \in \F_{p^n}$, the Boomerang Connectivity Table (BCT) entry of the $(n,n)$ function $F$ at point $(a,b)$, denoted by $\cB_f(a,b)$, is the number of solutions in $\F_{p^n} \times \F_{p^n}$ of the following system of equations
\begin{equation}\label{bs}
 \begin{cases}
 F(x)-F(y)=b,\\
F(x+a)-F(y+a)=b.
 \end{cases}
\end{equation}
The {\em boomerang uniformity} of the function $F$, denoted by $\cB_F$, is given by 
$$\cB_F = \max\{ \cB_F(a,b)\, :\, a,b \in\F_{p^n}^* \}.$$

For a Boolean or $p$-ary function $f:\F_{p^n}\to \F_p$, we define the {\em Walsh-Hadamard transform} to be the complex-valued function ($\zeta_p=e^{\frac{2\pi \, i}{p}}$ is a complex $p$th-root  of $1$)
\[
\cW_f(u) = \sum_{x\in \F_{p^n}}\zeta_p^{f(x)-\Trn(u x)}.
\]
  For an $(n,n)$-function $F$ and for $a,b \in \F_{p^n}$, we let the Walsh transform $\cW_F(a,b)$ of $F$ to be the Walsh-Hadamard transform of its component function ${\rm Tr}_1^n(bF(x))$ at $a$, that is,
\[
	\cW_F(a,b)=\sum_{x\in\F_{p^n}} \zeta_p^{\Trn(bF(x) - ax)}.
\]
 A {\em bent function} ($p$-ary or vectorial $(n,k)$; it is known that $k$ must satisfy $k\leq n/2$) is a function which has all of its absolute Walsh-Hadamard coefficients equal to $p^{n/2}$. A $p$-ary (or vectorial function) $f$ is called {\em plateaued} if
$|\mathcal{W}_f(\uu)| \in \{0,p^{(n+s)/2}\}$ for all $\uu\in \F_{p^n}$ for a fixed integer $s$ depending on $f$ (we also call $f$ then $s$-{\em plateaued}).
If $s=1$ ($n$ must then be odd), or $s=2$ ($n$ must then be even), we call $f$ {\em semibent}.

Investigating a practical attack on ciphers that use modular multiplication as a primitive operation, the authors of~\cite{BCJW02} used  a new differential  for a Boolean (vectorial) function $F$.
Drawing inspiration from the mentioned  successful attempt, along with P. Ellingsen, P. Felke and A. Tkachenko (see~\cite{EFRST20}),  two of us  defined a new (output) multiplicative differential, and the corresponding generalized differential uniformity.  
For a $p$-ary $(n,m)$-function   $F:\F_{p^n}\to \F_{p^m}$, and $c\in\F_{p^m}$, the ({\em multiplicative}) {\em $c$-derivative} of $F$ with respect to~$a \in \F_{p^n}$ is the  function
$
 _cD_{a}F(x) =  F(x + a)- cF(x), \mbox{ for  all }  x \in \F_{p^n}.
$
(Note that, if   $c=1$, then we obtain the usual derivative, which we denote by $D_a$, and, if $c=0$ or $a=0$, then we obtain a shift of the function.)

For an $(n,n)$-function $F$, and $a,b\in\F_{p^n}$, we let $_c\Delta_F(a,b)=\cardinality{\{x\in\F_{p^n} : F(x+a)-cF(x)=b\}}$ and
\[
_c\Delta_F=\max\left\{_c\Delta_F(a,b)\,:\, a,b\in \F_{p^n}, \text{ and } a\neq 0 \text{ if $c=1$} \right\}
\]
be the {\em $c$-differential uniformity} of~$F$. If $_c\Delta_F=\delta$, then we say that $F$ is differentially $(c,\delta)$-uniform. If $\delta=1$, then $F$ is  a {\em perfect $c$-nonlinear} ({\em PcN}) function (certainly, for $c=1$, they only exist for odd characteristic $p$; however, as proven in ~\cite{EFRST20}, there exist PcN functions for $p=2$, for all  $c\neq1$). If $\delta=2$, then $F$ is   an {\em almost perfect $c$-nonlinear} ({\em APcN}) function. 
When we specify the constant $c$ for which the function is PcN or APcN, then we may use the notation $c$-PN, or $c$-APN.
We note that if $F$ is an $(n,n)$-function, that is, $F:\F_{p^n}\to\F_{p^n}$, then $F$ is PcN if and only if $_cD_a F$ is a permutation polynomial.

In this paper we extend Dillon's switching method (see~\cite{Dillon09,EP09})  to $c$-differentials, and apply it to find necessary and sufficient conditions for such a constructed function to be PcN or APcN, as well as to generalize it to any $c$-differential uniformity.
A side note, but very important, is that, since the $c$-differential uniformity is not invariant under the CCZ-equivalence, an approach to improve the $c$-differential uniformity of a classical PN/APN function whose $c$-differential uniformity is not very good (like the Gold function) is to ``switch'' it via a Boolean/$p$-ary linearized function, and {\em decrease}, if possible, its $c$-differential uniformity while preserving its classical  differential uniformity. This theme occurs in some of our results (e.g., Example \ref{Ex-optimal}), though we push the method a lot further by also constructing other low $c$-differential functions from known ones.

 \section{Known classes of low $c$-differentially uniform functions}
 
 We include here two tables (also summarized in \cite{AaronThesis}) containing some of the known classes with low $c$-differential uniformity. We use $v_2$ as the $2$-valuation of the input, that is the largest power of $2$ dividing the input; the inverse is taken in the sense of modulo $p^n-1$ for the respective prime $p$, and the g.c.d. of two integers $r, s$ is denoted as $(r, s)$ for short. 
 Table~\ref{tab:long} lists monomials $x^d$ stating just the exponents~$d$. Table~\ref{tab2:long}
 lists the known polynomials (including those obtained in this paper) with low $c$-differential uniformity (here, $l > 1$ is a divisor of $p^n - 1$ and $g$ is a primitive element of $\F_{p^n}$, and $D_0$ is the multiplicative subgroup of $\F_{p^n}$ generated by $g$).
  \begin{center}
  \footnotesize{
   \begin{longtable}{|p{3cm}|c|p{2cm}|p{4.2cm}|c|}
   \caption{${_c}\Delta_F$ of various classes of functions $x^d$, $c \neq 1$} \label{tab:long} \\
   
   \hline \multicolumn{1}{|p{3cm}|}{$d$} & \multicolumn{1}{c|}{$\F_{p^n}$} & \multicolumn{1}{c|}{${_c}\Delta_F$} &  \multicolumn{1}{p{4.2cm}|}{Conditions} & \multicolumn{1}{c|}{Ref}\\ \hline 
\endfirsthead
   
   \multicolumn{5}{c}%
{{\tablename\ \thetable{} -- continued from previous page}} \\
   
 \hline \multicolumn{1}{|p{3cm}|}{$d$} & \multicolumn{1}{c|}{$\F_{p^n}$} & \multicolumn{1}{c|}{${_c}\Delta_F$} &  \multicolumn{1}{p{4.2cm}|}{Conditions} & \multicolumn{1}{c|}{Ref}\\ \hline 
\endhead
   
\hline \multicolumn{5}{|r|}{{Continued on next page}} \\ \hline
\endfoot   

\hline
\endlastfoot

     $2$ & $p>2$  & 2 (AP$c$N) & none &\cite{EFRST20}\\
              
       \hline
        $\frac{3^k+1}{2}$ & $p=3$ & 1 (P$c$N)   & $c=-1$, $\frac{2n}{\g(k,2n)}$ is odd&\cite{EFRST20} \\        
       \hline
       ${p^n-2}$ & any $p$ & 1 (P$c$N)    & $c=0$ &\cite{EFRST20} \\        
       \hline
         ${2^n-2}$ & $p=2$ & 2 (AP$c$N)    & $c \neq 0$, $\Trn(c)=\Trn(1/c)=1$ &\cite{EFRST20} \\        
       \hline
        ${2^n-2}$ & $p=2$ & 3  & $c \neq 0$, $\Trn(c)=0$ or $\Trn(1/c)=0$ &\cite{EFRST20} \\        
       \hline
 	${p^n-2}$ & $p>2$ & 2 (AP$c$N)    & $c \neq 0$, $(c^2-4c) \notin [\F_{p^n}]^2$, $(1-4c) \notin [\F_{p^n}]^2$, or $c=4,4^{-1}$ &\cite{EFRST20} \\        
       \hline
       ${p^n-2}$ & $p>2$ & 3   & $c\neq 0,4,4^{-1}$, $(c^2-4c) \in [\F_{p^n}]^2$ or $(1-4c) \in [\F_{p^n}]^2$ &\cite{EFRST20} \\        
       \hline
         ${2^k+1}$ & $p=2$ & $\frac{2^{\g(2k,n)}-1}{2^{\g(k,n)}-1}$ & $c \in \F_{2^{\g(n,k)}}\setminus \{1\}$, $\frac{n}{\g(n,k)} \geq 3 (n\geq3)$ &\cite{MRSZY21} \\   
         \hline
         ${2^k+1}$ & $p=2$ & $2^{\g(n,k)}+1$ & $c \in \F_{2^n}\setminus \F_{2^{\g(n,k)}}$ &\cite{MRSZY21} \\   
       \hline
   	${p^k+1}$ & any $p$ & $  \g(p^k+1,p^n-1)$ & $c \in \F_{p^{\g(n,k)}}$&\cite{MRSZY21} \\        
       \hline
        $\frac{p^k+1}{2}$ & $p>2$ &  $p^{\g(n,k)}+1$  & $c =-1$ &\cite{MRSZY21}\\        
       \hline
        $\frac{p^n+1}{2}$ & $p>2$ & $\leq 4$ & $c\neq \pm 1$&\cite{MRSZY21} \\  
        \hline
        $\frac{p^n+1}{2}$ & $p>2$ & $\leq 2$ & $c\neq \pm 1$, $\eta\big(\frac{1-c}{1+c}\big)=1$ $p^n\equiv 1 $ (mod 4)&\cite{MRSZY21} \\ 
       \hline
  	$\frac{2p^n-1}{3}$ & any & $\leq 3$   & $p^n \equiv 2\pmod 3$ &\cite{MRSZY21}\\  
       \hline
       $\frac{p^n+3}{2}$ & $p>3$ & $\leq 3$   & $c=-1$, $p^n \equiv 3\pmod 4$ &\cite{MRSZY21}\\   
       \hline
      $\frac{p^n+3}{2}$ & $p>3$ & $\leq 4$   & $c=-1$, $p^n \equiv 1\pmod 4$ &\cite{MRSZY21}\\  
       \hline
       $\frac{p^n-3}{2}$ & $p>2$ & $\leq 4$   & $c=-1$ &\cite{MRSZY21}\\        
       \hline
	$\frac{3^n+3}{2}$ & $p=3$ & 2 (AP$c$N)    & $c=-1$, $n$ even &\cite{MRSZY21}\\   
	\hline
	$\frac{3^n-3}{2}$ & $p=3$ & 6  & $c=-1$, $n=0\pmod 4$ &\cite{MRSZY21}\\ 
	\hline
	$\frac{3^n-3}{2}$ & $p=3$ & 4  & $c=-1$, $n\neq0\pmod 4$ &\cite{MRSZY21}\\ 
	\hline
	$\frac{3^n-3}{2}$ & $p=3$ & 2 (AP$c$N)  & $c=0$, &\cite{MRSZY21}\\ 
       \hline
       $\frac{3^n+1}{4}\, (\frac{3^k+1}{4})^{-1}$  & $p=3$ & 1 (P$c$N)   & $n,k$ odd, $c=-1$,  $\g(n,k)=1$  &\cite{ZH21}\\        
       \hline
         $\frac{5^n-1}{2} + (\frac{5^k+1}{2})^{-1} $ & $p=5$ & 1 (P$c$N)   & $n,k$ odd, $c=-1$,  $\g(n,k)=1$  &\cite{ZH21}\\        
       \hline
         $\frac{p^n+1}{2}\,(p^k+1)^{-1}$ & $p>2$ & $\leq 6$  & $d$ even, $c=-1$, $p^n \equiv 3 \pmod 4$  &\cite{ZH21}\\        
       \hline
       $\frac{p^n+1}{2}\, (p^k+1)^{-1}$ & $p>2$ & $\leq 3$  & $d$ odd, $c=-1$,  $p^n \equiv 3 \pmod 4$  &\cite{ZH21}\\        
       \hline
         $\frac{p^n+1}{4}+\frac{p^n-1}{2}$ & $p>2$ & $\leq 3$    & $c=-1$, $p^n \equiv 7 \pmod{8}$  &\cite{ZH21}\\ 
         	\hline
	${\frac{p^n-1}{2}+p^k+1}$ & $p>2$ & $\leq 3$    & $c=-1$, $\frac{n}{\g(n,k)}$ odd, $p^n \equiv 3 \pmod{4}$  &\cite{ZH21}\\ 
       \hline
	${\frac{p^n-1}{2}+p^k+1}$ & $p>2$ & $\leq 6$    & $c=-1$, $\frac{n}{\g(n,k)}$ odd, $p^n \equiv 1 \pmod{4}$  &\cite{ZH21}\\  

	 \hline
	${\frac{p^l+1}{2}}$ &  $p>2$ & 1 (P$c$N)    & $c=-1$, $l=0$ or $l$ even and $n$ odd, or $l,n$ both even together with $t_2 \geq t_1 +1$, where $n=2^{t_1u}$ and $l=2^{t_2}$ such that $2 \not |  u,v$&\cite{HPRS21}\\ 
       \hline
	${\frac{p^l+1}{2}}$ & $p>2$ & $\frac{p+1}{2}$    & $c=-1$, $\g(l,2n)=1$, $p\equiv 1\pmod{4}$ or $p\equiv 3 \pmod{8}$ &\cite{HPRS21}\\        
	 \hline
	${\frac{5^l+1}{2}}$ & $p=5$ & 3 & $c=-1$, $\g(l,2n)=1$ &\cite{HPRS21}\\        
	 \hline
	${\frac{3^l+1}{2}}$ & $p=3$ & 2 (AP$c$N) & $c=-1$, $\g(l,2n)=1$ &\cite{HPRS21}\\     
\hline
$p^4+(p-2)p^2$ +  $p(p-1)+1$ & $p>2$ & 1 (P$c$N)    & $c=-1$, $n=5$  &\cite{HPRS21}\\        
\hline
	${\frac{p^5+1}{p+1}}$ & $p>2$ & 1 (P$c$N) & $c=-1$, $n=5$ &\cite{HPRS21}\\
\hline
$(p-1)p^6+p^5+(p-2)p^3+(p-1)p^2+p$ & $p>2$ & 1 (P$c$N)    & $c=-1$, $n=7$  &\cite{HPRS21}\\
\hline
	${\frac{p^7+1}{p+1}}$ & $p>2$ & 1 (P$c$N) & $c=-1$, $n=7$ &\cite{HPRS21}\\
\hline
$x^{\frac{p^n+7}{2}}$ & $p=3$ & $\leq 2$ (AP$c$N) & $c=-1$, $n$ odd  &\cite{WLZ21}\\  
\hline
	$\frac{3^{\frac{n+1}{2}-1}}{2}$ & $p=3$ & $\leq 2$ (AP$c$N) & $c=-1$, $n\equiv 1 \pmod{4}$ &\cite{Yan21}\\ 
\hline
$\frac{3^{\frac{n+1}{2}-1}}{2}+\frac{3^n-1}{2}$ & $p=3$ & $\leq 2$ (AP$c$N) & $c=-1$, $n\equiv 3 \pmod{4}$ &\cite{Yan21}\\ 
\hline
$\frac{3^{n+1}-1}{8}$ & $p=3$ &  $\leq 2$ (AP$c$N) & $c=-1$, $n\equiv 1 \pmod{4}$ &\cite{Yan21}\\ 
\hline
$\frac{3^{n+1}-1}{8}+\frac{3^n-1}{2}$ & $p=3$ &  $\leq 2$ (AP$c$N) & $c=-1$, $n\equiv 3 \pmod{4}$ &\cite{Yan21}\\ 
\hline
$(3^{\frac{n+1}{4}}-1)(3^{\frac{n+1}{2}}+1)$ & $p=3$ &  $\leq 4$ & $c=-1$, $n\equiv 3 \pmod{4}$ &\cite{Yan21}\\ 
\hline
$\frac{3^n+1}{4}+\frac{3^n-1}{2}$ & $p=3$ &  $\leq 4$ & $c=-1$, $n$ odd &\cite{Yan21}\\ 
\hline
${d^{-1}}\pmod{p^n-1}$ & any $p$ & 1 (P$c'$N)  & $x^d$ is P$c$N, $c'=c^d$ &\cite{WZ21}\\ 
\hline
$\{2^j\}_{j\geq 0}, \{2^j(2^k+1)\}_{k,j\geq 0}$ & $p=2$ & 1 (P$c$N)  &  $c\neq 1$, resp., $c\in\F_{2^{(k,n)}}\setminus\{1\},v_2(n)\leq v_2(k)$ &\cite{WZ21}\\ 
\hline
odd $2 (p^k+1)^{-1}\pmod{p^n-1},k\geq 0$ & $p>2$ & 1 (P$c$N)  & $c=-1$ &\cite{WZ21}\\
\hline
$\frac{p^n+1}{2}\left(\frac{p^k+1}{2} \right)^{-1}  $ & $p>2$ & 1 (P$c$N)  & $c=-1$, $v_2(k)=v_2(n)$, \hspace{8mm} $p^n\equiv 1 \pmod{4}$ &\cite{WZ21}\\
\hline
         \end{longtable}
        }
            \end{center}
        
         \begin{center}
        	\scriptsize{
        		\begin{longtable}{|p{5cm}|c|p{1.55cm}|p{3cm}|c|}
        			\caption{${_c}\Delta_F$ of various classes of functions $F(x)$, $c \neq 1$} \label{tab2:long} \\
        			
        			\hline \multicolumn{1}{|p{5.2cm}|}{$F(x)$} & \multicolumn{1}{c|}{$\F_{p^n}$} & \multicolumn{1}{c|}{${_c}\Delta_F$} &  \multicolumn{1}{p{3cm}|}{Conditions} & \multicolumn{1}{c|}{Ref}\\ \hline 
        			\endfirsthead
        			
        			\multicolumn{5}{c}%
        			{{\tablename\ \thetable{} -- continued from previous page}} \\
        			
        			\hline \multicolumn{1}{|p{5.2cm}|}{$F(x)$} & \multicolumn{1}{c|}{$\F_{p^n}$} & \multicolumn{1}{c|}{${_c}\Delta_F$} &  \multicolumn{1}{p{3cm}|}{Conditions} & \multicolumn{1}{c|}{Ref}\\ \hline 
        			\endhead
        			
        			\hline \multicolumn{5}{|r|}{{Continued on next page}} \\ \hline
        			\endfoot   
        			
        			\hline
        			\endlastfoot
        			
        			   $x^{10} - ux^6 - u^2x^2$ & $p=3$ & $\geq 2$ & $u\in \F_{3^n}$&\cite{EFRST20} \\ 
        			\hline
        			$L(x)(\sum_{i=1}^{l-1} L(x)^{\frac{p^n-1}{l}i}+u)$ & any $p$ & $\leq 2$ (AP$c$N) & $L$ an $\F_p$-linearized polynomial,  $l|(p^n-1)$, $u \neq 1,(1-l) \mod p$, $1-\frac{l}{(1-c)(u+l-1)},1+\frac{l}{(1-c)(u-1)} \in D_0$ &\cite{WLZ21}\\  
        			\hline
        			$(x^{p^k}-x)^{\frac{q-1}{2}+1}+a_1x+a_2x^{p^k}+a_3x^{p^{2k}}$ & $p=3$ & $\leq 2$ (AP$c$N) & $c=-1$, $0\leq i\leq2$, $a_1,a_2,a_3\in \F_3$, $a_1+a_2+a_3\neq 0$  &\cite{WLZ21}\\  
        			\hline    			
        			$f(x)(\Trn(x)+1)+f(x+\gamma)\Trn(x)$ & $p=2$ & 1 (P$c$N) &  $f(x)$ is P$c$N, $\gamma \in \F_{p^n}^*$  &\cite{WLZ21}\\  
        			\hline
        			$L(x)+L(\gamma)(\Trn(x))^{q-1}$ &any $p$ & 1 (P$c$N) & $L$ an $\F_q$-linearized polynomial, $\gamma \in \F_{q}^*$, $\Trn(\gamma) =0$  &\cite{WLZ21}\\  
        			\hline
        			$u\phi(x)+g((\Trn(x))^q)-g(\Trn(x))$ &any $p$ & 1 (P$c$N) & $\phi$ an $\F_q$-linearized polynomial, $u\in \F_q^*$, ker($\phi$)$\cap$ ker($\Trn$)=$\{0\}$, $g \in \F_{q^n}[x]$  &\cite{WLZ21}\\ 
        			\hline
        			$u(x^q-x)+g(\Trn(x))$ &any $p$ & 1 (P$c$N) & $g \in \F_{q^n}[x]$ a permutation of $\F_q$, $u \in \F_{q}^*$, $p \nmid n$  &\cite{WLZ21}\\ 
        			\hline
        			$F(x)+u\Trn(vF(x))$ & any $p$ & 1 (P$c$N)  & $F$ is P$c$N, $\Trn(-uv)\neq 1$ & here\\
        			\hline
        			$L_1(x)+L_1(\gamma)\Trn(L_2(x))$ & any $p$  &  1 (P$c$N)   &   $\Trn\left(\frac{L_1(\gamma)}{1-c} \right)=0$, $\Trn(\gamma)=0$ & here\\
        			\hline
        			$L(x)+\prod_{i=1}^s\left(\Trn(x^{2^{k_i}+1}+\delta_i)\right)^{g_i}$ & $p=2$ &  $\leq 2$ (AP$c$N) & $1\leq k_i \leq n-1$ & here\\
        			\hline
        			$L(x)+\prod_{i=1}^s\left(\alpha_i\Tr_{q^n/q^m}(x^{2^{k_i}+1}+\delta_i)\right)^{g_i}$ & $p=2$ &  $\leq 2$ (AP$c$N) & $g_i\geq 1,\,\delta_i\in\F_{2^n}$, $\alpha_i\in\F_{2^m}^*$, $1\leq k_i \leq n-1$ & here\\
        			\hline
        			$L(x)+u\sum_{i=1}^t\left(\Tr_{q^n/q^m}(x)^{k_i}+\delta_i \right)^{s_i}$ & any $p$ & P$c$N & 
        			$pm\,|\,n$, $1\leq t\in\mathbb{Z}_{>0}$, $u\in\F_{p^m}^*$, $\delta_i\in\F_{p^m}$, 
        			$1\leq k_i,s_i\leq p^n-1$, $L$ linearized permutation, $c\in \F_{p^m}\setminus\{1\}$ & here\\
        			\hline
  \end{longtable}
}
\end{center}
 \section{The $c$-switching method}
 
 We first recall the {\em switching method} introduced by Dillon~\cite{Dillon09} and pushed further by Edel and Pott~\cite{EP09} (there are several proofs of this result besides the original one and we point to~\cite{Roberto20} for a very detailed argument).
 
 \vskip.35cm
 \noindent
{\bf Dillon's Switching Method.}
 {\em Let $F:\F_{2^n}\to \F_{2^n}$ be an APN function, $u\in\F_{2^n}^*$, $f:\F_{2^n}\to\F_2$, a Boolean function, and $H(x)=F(x)+u\,f(x)$. Then $H$ is an APN function if and only if 
 $D_a f(x)+D_a f(y)=0$, whenever $D_a F(x)+D_aF(y)=u$, for all $a\neq 0,x,y\in\F_{2^n}$.
}
 
 \medskip
 
 Below, we generalize the switching method to characterize the PcN functions constructed by changing only some components. We will use the univariate representation, as it is more convenient in this context. We write the theorem for any characteristic~$p$. Throughout this paper, $q=p^t$, where $p$ is a prime and $t>0,n>2$ are integers.
 \begin{thm}
 \label{thm:pcn}
 Let $u\in\F_{q^n}^*$, $c \in\F_{q^n}$ ($c\neq1$ if $p=2$), $F:\F_{q^n}\rightarrow\F_{q^n}$ a PcN function, and $f:\F_{q^n}\rightarrow\F_q$. We define $H^{(1)}:\F_{q^n}\rightarrow\F_{q^n}$ by $H^{(1)}(x)=F(x)+u\, f(x)$. Then $H^{(1)}$ is not PcN if and only if there exist at least two $x\neq y$ such that ${_c}D_a f(x)-{_c}D_af(y)=\epsilon$, whenever ${_c}D_a F(x)-{_c}D_aF(y)=-\epsilon\, u$, where $\epsilon\in\{\alpha-c\beta\, : \, \alpha,\beta\in\F_{q},\alpha-c\beta\neq0\}$.
 
 More generally, let $u_i\in\F_{q^n}^*$, $1\leq i\leq k$, be a PcN function  $F:\F_{q^n}\rightarrow\F_{q^n}$, $f_i:\F_{q^n}\rightarrow\F_q$, and $H^{(k)}(x)=F(x)+\sum_{i=1}^k u_i f_i(x)$. Then $H^{(k)}$ is not PcN if and only if, for some $a\in\F_{q^n}$, there exist $x\neq y$ such that ${_c}D_a f_i(x)-{_c}D_a f_i(y)=\epsilon_i$, $1\leq i\leq k$, whenever ${_c}D_a F(x)-{_c}D_a F(y)=-\sum_{i=1}^k u_i\epsilon_i$, $\epsilon_i\in\{\alpha-c\beta\,:\, \alpha,\beta\in\F_q\}$ (not all $\epsilon_i$ are zero).
 \end{thm}
  \begin{proof}
First, we observe that the $c$-differential equation for $H^{(1)}$ at $a,b$ is
\begin{equation}
\label{eq:main}
F(z+a)-cF(z)+(f(z+a)-cf(z)) u=b.
\end{equation}
Since $F$ is PcN, this equation will have at most $q^2$ solutions, namely the ones such that 
\[
F(z+a)-cF(z)\in\{b-(\alpha-c\beta)u,\, \alpha,\beta\in\F_{q}\}.
\]

NB: Note that, here, $\alpha-c\beta$ can be zero. Note also that, for $c\in\F_q$, we have at most $q$ solutions.

Suppose that $H^{(1)}$ is not PcN. Let $x,y$ be two such solutions of~\eqref{eq:main} for some $b$. Thus,
\begin{align*}
F(x+a)-cF(x)+(f(x+a)-cf(x)) u&=b\\ 
F(y+a)-cF(y)+(f(y+a)-cf(y)) u&=b.
\end{align*}
Let $f(x+a)-cf(x)=\alpha-c\beta$, $f(y+a)-cf(y)=\alpha'-c\beta'$. Observe that $f(x+a)-cf(x)$ cannot equal $f(y+a)-cf(y)$, since then (with $\delta=(f(x+a)-cf(x))u$), the equation $F(x+a)-cF(x)=b-\delta$ would have two solutions, $x,y$, and so, $F$ would not be PcN. Thus, $(\alpha-\alpha')-c(\beta-\beta')\neq0$. We have that 
\begin{align*}
F(x+a)-cF(x)&=b-(\alpha-c\beta)u\\
F(y+a)-cF(y)&=b-(\alpha'-c\beta')u,
\end{align*}
and so, $_cD_aF(x)-{_c}D_aF(y)=-u(\alpha-\alpha'-c(\beta-\beta'))$ and $_cD_af(x)-{_c}D_af(y)=\alpha-\alpha'-c(\beta-\beta')$.

Conversely, we assume that $_cD_a f(x)-{_c}D_af(y)=\epsilon$ and $_cD_a F(x)-{_c}D_aF(y)=-\epsilon\, u$, for some $x\neq y$ and some $a$, and want to show that $H^{(1)}$ is not PcN. We let 
\[
\gamma_1={_c}D_aF(x),\ \gamma_2={_c}D_af(x)\, u,
\]
and take $b=\gamma_1+\gamma_2$. We shall show that the $c$-differential equation for  $H^{(1)}$ at $a,b$ will have at least two solutions and hence $H^{(1)}$ cannot be PcN.
From the imposed conditions,  we infer that
\[
{_c}D_aF(y)=\epsilon\, u+\gamma_1,\ {_c}D_af(y)=-\epsilon+\frac{\gamma_2}{u}.
\]
We thus get
\begin{align*}
&{_c}D_aF(x)+ {_c}D_af(x) \,u=\gamma_2+\gamma_1=b,\\
&{_c}D_aF(y)+ {_c}D_af(y) \,u=\epsilon\, u+\gamma_1+\left(-\epsilon +\frac{\gamma_2}{u}\right)u=\gamma_1+\gamma_2=b,
\end{align*}
and so, $H^{(1)}$ is not PcN.

The general case follows also similarly.
If $H^{(k)}$ is not PcN, then, for some $a,b\in\F_{q^n}$, there exist $x\neq y$ such that
\begin{align*}
&{_c}D_a F(x)+\sum_{i=1}^k u_i\cdot  {_c}D_a f_i(x)=b\\
&{_c}D_a F(y)+\sum_{i=1}^k u_i\cdot  {_c}D_a f_i(y)=b.
\end{align*}
We now take $D_a f_i(x)=\alpha_i-c\beta_i$, and $D_a f_i(y)=\alpha_i'-c\beta_i'$. 
As we observed earlier, since $F$ is PcN, then $\sum_{i=1}^k u_i\cdot  {_c}D_a f_i(x)\neq \sum_{i=1}^k u_i\cdot  {_c}D_a f_i(y)$.
Moreover, 
\begin{align*}
{_c}D_a F(x)-{_c}D_a F(y)&=-\sum_{i=1}^k u_i ((\alpha_i-\alpha_i')-c(\beta_i-\beta_i'))
\end{align*}
while 
$
 {_c}D_a f_i(x)- {_c}D_a f_i(y)=(\alpha_i-\alpha_i')-c(\beta_i-\beta_i').
$
Conversely, we assume
\begin{align*}
{_c}D_a F(x)-{_c}D_a F(y)&=-\sum_{i=1}^k u_i\epsilon_i\\
 {_c}D_a f_i(x)- {_c}D_a f_i(y)&=\epsilon_i.
\end{align*}
Denoting ${_c}D_aF(x)=\alpha, {_c}D_a f_i(x) u_i=\lambda_i$, then
\begin{align*}
&{_c}D_a F(x)+\sum_{i=1}^k u_i\  {_c}D_a f_i(x)=\alpha+\sum_{i=1}^k \lambda_i\\
&{_c}D_a F(y)+\sum_{i=1}^k u_i\  {_c}D_a f_i(y)\\
&\qquad\qquad=\alpha+\sum_{i=1}^k u_i\epsilon_i+\sum_{i=1}^k u_i\left(\frac{\lambda_i}{u_i}-\epsilon_i\right)\\
&\qquad\qquad=\alpha+\sum_{i=1}^k \lambda_i,
\end{align*}
but that is impossible if $F$ is PcN.
\end{proof}

\begin{exa}
We can give an example of this type of function, like $H_2(x)=x+\Tr_4(x)$, which is PcN on $\F_{2^4}$, for all $c\neq 1$, where $g$ is a primitive element of $\F_{2^4}$.
\end{exa}

 \begin{rem}
 We can rewrite the previous theorem in a ``positive'' manner, by describing the PcN property in lieu of the negation. For example,  $H^{(1)}$ is PcN if and only if  $_cD_a f(x)-{_c}D_af(y)\neq \epsilon$, whenever ${_c}D_a F(x) - {_c}D_aF(y)=-\epsilon\, u$, where $\epsilon\in\{\alpha-c\beta\, : \, \alpha,\beta\in\F_{q},\alpha-c\beta\neq0\}$.
 \end{rem}
\begin{rem}\label{max}
Note that the proof also implies that $H^{(1)}$ has   $c$-differential uniformity at most $q^2$ for $c\not\in \F_q$ and at most $q$ for $c \in \F_q$.
\end{rem}

  \begin{rem}
 \label{thm:pcngen}
 The function $f$ can be chosen with outputs in any subfield of $\F_{q^n}$ and the result in the theorem above is still true.
 Let $u\in\F_{q^n}^*$, $c \in\F_{q^n}$ ($c\neq1$ if $p=2$), $F:\F_{q^n}\rightarrow\F_{q^n}$ a PcN function, and $f:\F_{q^n}\rightarrow\F_{q^m},m\,|\,n$. We denote the image of $f$ by $Im(f)$. Then $H^{(1)}:\F_{q^n}\rightarrow\F_{q^n}$ given by $H^{(1)}(x)=F(x)+u\, f(x)$  is not PcN if and only if there exist at least two $x\neq y$ such that ${_c}D_a f(x)-{_c}D_af(y)=\epsilon$, whenever ${_c}D_a F(x)-{_c}D_aF(y)=-\epsilon\, u$, where $\epsilon\in\{\alpha-c\beta\, : \, \alpha,\beta\in Im(f),\alpha-c\beta\neq0\}$.
\end{rem}
  
It is rather interesting that we can easily construct new PcN functions from old ones, via our Theorem~\ref{thm:pcn}, and we record that construction below. In addition, we generalize~\cite[Theorem 5]{WLZ21}.
\begin{thm}
Let $q=p^t$ be a power of a prime $p$, $n>2$, and $F:\F_{q^n}\to\F_{q^n}$ be a PcN function for some $c\in\F_q$, and let $u,v\in\F_{q^n}$ with $\Trn(-uv)\neq 1$ (recall that $\Tr_{q^n/q}=\Trn$). Then 
\[
H(x)=F(x)+u\Trn(vF(x))
\]
 is PcN with respect to $c$. Furthermore, 
let $L_1,L_2\in\F_q[x]$ be linearized permutation polynomials over $\F_{q^n}$, $\gamma\in\F_{q^n}^*$, such that $\Trn(\gamma)=0$. Then \[
H(x)=L_1(x)+L_1(\gamma)\Trn(L_2(x))
\]
 is PcN with respect to $c\in\F_{q^n}\setminus\{1\}$.
\end{thm} 
 \begin{proof}
 Let $f(x)=\Trn(vF(x))$.
 Assume that there exist $x\neq y$ such that ${_c}D_a F(x) - {_c}D_aF(y)=-u(\alpha-c\beta)$ and $_cD_a f(x)-{_c}D_af(y)=\alpha-c\beta\neq 0$, $\alpha,\beta\in\F_q$. The first equation implies that
 \begin{align*}
 \Trn\left(v\left({_c}D_a F(x) - {_c}D_aF(y) \right)\right)=\Trn\left(-uv(\alpha-c\beta) \right)
 \end{align*}
 and so,
  \begin{align*}
_cD_a f(x)-{_c}D_af(y)&= \alpha-c\beta=\Trn\left(-uv(\alpha-c\beta) \right)=(\alpha-c\beta)\Trn(-uv).
 \end{align*}
 Since $\Trn(-uv)\neq 1$, we arrive at a contradiction and the first  claim is shown.
 
 Now, for the second claim, we shall use below that $\Trn(L(x))=L(\Trn(x))$, for a linearized polynomial $L$ whose coefficients are in $\F_q$. We
 let $F(x)=L_1(x), f(x)=\Trn(L_2(x))$ and apply Theorem~\ref{thm:pcn}, by assuming that there exist $x\neq y$ such that
 \begin{align*}
 {_c}D_a F(x) - {_c}D_a F(y)&=(1-c) L_1(x-y)=-\epsilon L_1(\gamma)\\
 {_c}D_a f(x) - {_c}D_a f(y)&=(1-c) \Trn(L_2(x-y))\\
 &=(1-c)L_2(\Trn(x-y))=
\epsilon.
 \end{align*}

The second equation implies that $\Trn(L_2(x-y))=\frac{\epsilon}{1-c}$. Since $\Trn(z)\in\F_q$ for all $z$, this implies that $\frac{\epsilon}{1-c}\in\F_q$. Now, the first equation gives that $L_1(x-y)=\frac{-\epsilon L_1(\gamma)}{1-c}=L_1\left(\frac{-\epsilon \gamma}{1-c}\right)$. Using that $L_1$ is a permutation, we get $x-y=\frac{-\epsilon\gamma}{1-c}$. Using this again in the condition for $L_2$, we obtain $\frac{\epsilon}{1-c}= L_2(\Trn(x-y))=L_2\left(\Trn\left(\frac{-\epsilon\gamma}{1-c} \right)\right)=\frac{-\epsilon}{1-c}L_2(\Trn(\gamma))=0$, which is a contradiction. The theorem is shown.
 \end{proof}
 \begin{exa}
 For example, we obtain that $x^{\frac{3^k+1}{2}}+u\Trn\left(vx^{\frac{3^k+1}{2}}\right)$ is PcN on $\F_{3^n}$ for $c=-1$,  when $\Trn(-uv)\neq 1$,  $\gcd(k, n) = 1$ and $n$ is odd. A more concrete example is $x^{5}+g^2\Tr_3(gx^{5})$ on $\F_{3^3}$ ($g$ is a primitive element of $\F_{3^3}$), which is a $4$-differentially uniform (with respect to $c=1$) permutation and PcN (with respect to $c=-1$). Even using the inverse function and modifying one of its coordinates as in the theorem above, provides functions with low differential uniformity, for some cases, though we were not able to find general classes.
 \end{exa}

 \begin{exa}\label{Ex-optimal}
 There are some functions one gets via our method that are close to or even optimal, from many cryptographic perspectives. Moreover, the $c$-differential properties of the building function $F$ can be improved by switching.
 For example, let $n=6$ and $w$ be a primitive element of $\F_{2^6}$ satisfying $w^6+w^4+w^3+w + 1=0$. Take $F(x)=x^5$. This function is of Gold type and it is known that $F$ is a $4$-differentially uniform permutation. It has boomerang uniformity $4$ and  nonlinearity $24$  (optimal, the so-called bent concatenation bound). In addition, according to Table~\textup{\ref{tab:long}}, it is PcN for $c \in \{0, w^{21}, w^{42}\}$ and has $c$-differential uniformity $5$ for other values of $c$. Now we take $H(x)=x^5 + \Tr_{2^6/2^2}(vx^5) $, where $v$ is an element satisfying $\Tr_{2^6/2^2}(v) \neq 1$. 
 Experimental result shows that the function $H$ is a $4$-differentially uniform permutation with the same Walsh spectrum and boomerang uniformity as $F$. More importantly,  
 when we take $v$ from the following set $$\{ w^{11},w^{22},w^{44},w^{25}, w^{50}, w^{37}, w^{21}, w^{42}, w^{23}, w^{46},  w^{29}, w^{58}, w^{53}, w^{43}\},$$
 the function $H$ has $c$-differential uniformity $4$ for all the elements $c$ where $F$ has $c$-differential uniformity $5$. For the other elements $v$ satisfying   $\Tr_{2^6/2^2}(v) \neq 1$, the corresponding function $H$ has the same $c$-differential uniformity as $F$. This 
 demonstrates that  the switching method can be applied to improve the $c$-differential uniformity of some functions $F$ while preserving other good cryptographic properties of $F$.
 \end{exa}

 We can further generalize Theorem~\ref{thm:pcn} to any $c$-differential uniformity.  
 \begin{thm}
 \label{thm:gen} 
Let $u\in\F_{q^n}^*$ and $c \in\F_{q^n}$ (if $p=2$, then $c\neq 1$), $F:\F_{q^n}\rightarrow\F_{q^n}$ be a (at most) $(c,\delta)$-uniform function, $f:\F_{q^n}\rightarrow\F_q$, and $H^{(1)}(x)=F(x)+u\, f(x)$. For $a\in\F_{q^n}$, let $A_{a,\epsilon}=\{ x\,:\, f(x+a)-cf(x)=\epsilon\}$, where $\epsilon\in\{\alpha-c\beta\,: \, \alpha,\beta\in\F_{q},\alpha-c\beta\neq0\}$.
Then $H^{(1)}$ has $c$-differential uniformity  $_c\Delta_{H^{(1)}}>\delta$ if and only if  there exist $x_1,x_2, \dots, x_{\delta+1}$  (not all belonging to the same set $A_{a,\epsilon}$) such that, for all $i\neq j$, $a\in\F_{q^n}$, if ${_c}D_a F(x_i)-{_c}D_aF(x_j)=-\epsilon\, u$ then  ${_c}D_a f(x_i)-{_c}D_af(x_j)=\epsilon$, where $\epsilon\in\{\alpha-c\beta\,: \, \alpha,\beta\in\F_{q},\alpha-c\beta\neq0\}$.

 More generally, let $u_i\in\F_{q^n}^*$, $1\leq i\leq k$, $F:\F_{q^n}\rightarrow\F_{q^n}$ be a (at most) $(c,\delta)$-uniform function, $f_t:\F_{q^n}\rightarrow\F_q$, and $H^{(k)}(x)=F(x)+\sum_{t=1}^k u_t f_t(x)$. Then $H^{(k)}$ has $c$-differential uniformity  $_c\Delta_{H^{(k)}}>\delta$ if and only if, for some $a\in\F_{q^n}$, there exist $x_1,x_2, \dots, x_{\delta+1}$  (not all belonging to the same set $A_{a,\epsilon}$) such that, for all $i\neq j$, $a\in\F_{q^n}$, if ${_c}D_a F(x_i)-{_c}D_aF(x_j)=-\sum_{t=1}^k u_t\epsilon_t\,$ then  ${_c}D_a f_t(x_i)-{_c}D_af_t(x_j)=\epsilon_t$, where $\epsilon_t\in\{\alpha-c\beta\,: \, \alpha,\beta\in\F_{q},\alpha-c\beta\neq0\}$.
 \end{thm}
\begin{proof}
We assume first that $H^{(1)}$ has $c$-differential uniformity  $_c\Delta_{H^{(1)}}>\delta$. Then, there exist   elements $x_i$, $1\leq i\leq \delta+1$, such that, for some $a,b\in\F_{q^n}$,
\begin{align*}
F(x_i+a)-cF(x_i)+(f(x_i+a)-cf(x_i))u={_c}D_aF(x_i)+ {_c}D_a f(x_i)u=b.
\end{align*}
Running the above equation for two values $1\leq i\neq j\leq \delta+1$ and subtracting them,  we get
\begin{align*}
{_c}D_aF(x_i)-{_c}D_aF(x_j)+ ({_c}D_a f(x_i)-{_c}D_a f(x_j))u=0.
\end{align*}
Since $f$ has values in $\F_q$, then ${_c}D_a f(x_i)=\alpha-c\beta, {_c}D_a f(x_j)=\alpha'-c\beta'$, $\alpha,\alpha',\beta,\beta'\in\F_q$. From the above equation, we obtain that ${_c}D_aF(x_i)-{_c}D_aF(x_j)=-u(\alpha-\alpha'-c(\beta-\beta'))$, while ${_c}D_a f(x_i)-{_c}D_a f(x_j)=\alpha-\alpha'-c(\beta-\beta')$. Moreover, if all $x_i$ belong to the same set $A_{a,\epsilon}$, then the equation $F(z+a)-cF(z)=b-\epsilon u$ would have at least $\delta+1$ solutions, which is not allowed. 

Conversely, we assume that there exist $x_i$, $1\leq i\leq \delta+1$, such that ${_c}D_aF(x_i)-{_c}D_aF(x_j)=-\epsilon u$ and ${_c}D_a f(x_i)-{_c}D_a f(x_j)=\epsilon$, for some $\epsilon\in\{\alpha-c\beta\neq 0\,: \, \alpha,\beta\in\F_{q}\}$. As in the proof of Theorem~\ref{thm:pcn}, letting  
$
\gamma_1={_c}D_aF(x_1),\ \gamma_2={_c}D_af(x_1)\, u,
$
and taking $b=\gamma_1+\gamma_2$, then 
$
{_c}D_aF(x_j)=\epsilon\, u+\gamma_1,\ {_c}D_af(x_j)=-\epsilon+\frac{\gamma_2}{u},
$
$2\leq j\leq \delta+1$,
we get
\begin{align*}
&{_c}D_aF(x_1)+ {_c}D_af(x_1) \,u= b,\\
&{_c}D_aF(x_j)+ {_c}D_af(x_j) \,u=b, 2\leq j\leq \delta+1,
\end{align*}
so $H^{(1)}$ has $c$-differential uniformity at least $\delta+1$.

We assume now that $H^{(k)}$ has $c$-differential uniformity  $_c\Delta_{H^{(k)}}>\delta$. Then, there exist   elements $x_i$, $1\leq i\leq \delta+1$, such that, for some $a,b\in\F_{q^n}$,
\begin{align*}
&F(x_i+a)-cF(x_i)+\sum_{t=1}^ku_t(f_t(x_i+a)-cf_t(x_i))\\
&\qquad={_c}D_aF(x_i)+ \sum_{t=1}^ku_t\cdot{_c}D_a f_t(x_i)=b.
\end{align*}
Running the above equation for two values $1\leq i\neq j\leq \delta+1$ and subtracting them,  we get
\begin{align*}
{_c}D_aF(x_i)-{_c}D_aF(x_j)+\sum_{t=1}^k ({_c}D_a f_t(x_i)-{_c}D_a f_t(x_j))u_t=0.
\end{align*}
Since $f$ has values in $\F_q$, then ${_c}D_a f_t(x_i)=\alpha_t-c\beta_t, {_c}D_a f_t(x_j)=\alpha_t'-c\beta'_t$, $\alpha_t,\alpha'_t,\beta_t,\beta'_t\in\F_q$. From the above equation, we obtain that ${_c}D_aF(x_i)-{_c}D_aF(x_j)=-\sum_{t=1}^k u_t(\alpha_t-\alpha'_t-c(\beta_t-\beta'_t))$, while ${_c}D_a f_t(x_i)-{_c}D_a f_t(x_j)=\alpha_t-\alpha'_t-c(\beta_t-\beta'_t)$. Moreover, if, for every $t$, all $x_i$ belong to the same set $A_{a,\epsilon_t}$, then the equation $F(z+a)-cF(z)=b-\sum_{t=1}^k\epsilon_t u_t$ would have at least $\delta+1$ solutions, which is not allowed. 

Conversely, we assume that there exist $x_i$, $1\leq i\leq \delta+1$, such that ${_c}D_aF(x_i)-{_c}D_aF(x_j)=-\sum_{t=1}^k\epsilon_t u_t$ and ${_c}D_a f_t(x_i)-{_c}D_a f_t(x_j)=\epsilon_t$, for some $\epsilon_t\in\{\alpha-c\beta\neq 0\,: \, \alpha,\beta\in\F_{q}\}$. As in the proof of Theorem~\ref{thm:pcn}, letting  
$
\gamma={_c}D_aF(x_1),\ \gamma_t={_c}D_af_t(x_1)\, u_t,
$
and taking $b=\gamma+\sum_{t=1}^k\gamma_t$, then 
$
{_c}D_aF(x_j)=-\sum_{t=1}^k\epsilon_t u_t+\gamma+\sum_{t=1}^k u_t\left(\frac{\gamma_t}{u_t}-\epsilon_t\right)=b,
$
$2\leq j\leq \delta+1$,
we get
\begin{align*}
&{_c}D_aF(x_1)+\sum_{t=1}^k {_c}D_af_t(x_1) \,u_t= b,\\
&{_c}D_aF(x_j)+ \sum_{t=1}^k{_c}D_af_t(x_j) \,u_t=b, 2\leq j\leq \delta+1,
\end{align*}
so $H^{(k)}$ has $c$-differential uniformity at least $\delta+1$.
\end{proof}

\begin{exa}
 We can easily find  examples of functions with good $c$-differential uniformity and good cryptographic properties for odd characteristics~$p$ using our theorems above. We give below just some random examples found via Magma~\textup{\cite{Magma}} \textup{(}notation: PN=perfect nonlinear,  cDU=$c$-differential uniformity, $g$ is a primitive element in the respective field\textup{):}\newline
 For $p=3$:
 \begin{enumerate}
 \item 
 for $n=5$, the function $x^{p^2+1}+\Tr(g^2 x^{p^2+1})$
 is PN ($c=1$), APcN for $c=-1$, and $6\leq cDU\leq 8$ for the remainding ones, for $c\neq 1$, and its Walsh spectrum is $6$-valued; for $n=6$, it is $3$-differentially uniform ($c=1$), $6\leq cDU\leq 11$, for $c\neq 1$, and has $8$-valued Walsh spectrum.
\item for $n=4,5,6$, $x^{\frac{p^2+3}{2}}+g \Tr(g^ix)$ ($i=0,1$)  is PN ($c=1$) and the $cDU=4$ for all other $c$'s, and has $6$-valued Walsh spectrum;
\item the function $x^{\frac{p^3+1}{2}}+g^2\Tr(gx)$ is PN with $6$-valued Walsh spectrum, and for  $c\neq 1$ it has $4\leq cDU\leq 8$ when $n=4$, $6\leq cDU\leq 10$ when $n=5$, respectively;
\item for $n=4, 5, 6$, the function $x^{\frac{p+1}{2}} + g\Tr(x)$ is PN, $cDU=4$, for $c\neq 1$,  and has $6$-valued Walsh spectrum.
 \end{enumerate}
  For $p=5$:
   \begin{enumerate}
 \item $x^{\frac{p^3+1}{2}}+g^2 \Tr(g^i x)$, $i=0,1$, are APN; for $n=2$ have cDU=$5$ and $8$ valued-Walsh spectrum; for $n=4$, APN, $7\leq cDU\leq 11$,  for $c\neq 1$ and $9$-valued Walsh spectrum;
\item $x^{\frac{p+1}{2}}+g \Tr(x)$ is  APN, $cDU=5$, $8$-valued Walsh spectrum, for $n=2$; for $n=3$, APN, $9$-valued Walsh spectrum, $cDU\in \{5,7\}$; for $n=4$, APN, $cDU\in\{8,9\}$, $9$-valued Walsh spectrum.
  \end{enumerate}
 \end{exa}

 We will use below a $q$-linearized polynomial $L(x)=\sum_{i=0}^{n-1} a_i x^{q^i}\in\F_{q^n}$, which is known (see, \cite{ZWW20}, for instance) to be a permutation polynomial if and only if the associated Dickson matrix 
 \[
 D_L=\begin{pmatrix}
 a_0 & a_1 &\cdots & a_{n-1}\\
 a_{n-1}^q & a_0^q &\cdots & a_{n-2}^q\\
 \vdots & \vdots &\vdots & \vdots\\
 a_1^{q^{n-1}} & a_2^{q^{n-1}} & \cdots & a_0^{q^{n-1}}
 \end{pmatrix}
 \]
 is nonsingular. In particular, one can take below the linearized polynomial $L(x)=x^{q^r}-ax$ over $F_{q^n}$, $1\leq r\leq n-1$, which is known to be a permutation polynomial if and only if $N_{q^n/q^d} (a)\neq 1$, $d=\gcd(n,r)$ (recall that the field norm $N$ is the product of all conjugates of the input). We have the following corollary, which will be frequently used in the remaining of our paper.
\begin{cor}\label{cor}
With the notation of Theorem~\textup{\ref{thm:gen}}, if $F$ is PcN, then the elements $x_i$ under those conditions belong to different sets $A_{a,\epsilon}$, and the $c$-differential uniformity of $H$ is at most $q^2-1$, and $q-1$ for $c\in\F_q$.
When $F(x)=\sum_{i=0}^{n-1} a_i x^{q^i}\in\F_{q^n}[x]$ is a linearized permutation polynomial (note that $F$ is PcN for all $c\neq1$), 
if we assume that the $c$-differential uniformity of $H$ ($c\neq1$) is $q^2-1$ ($q-1$ for $c\in\F_q\setminus\{1\}$), then, for every  $\gamma\in\F_q^*$, there exists a pair $\{x_i,x_j\}$ such that {$(1-c)F(x_i-x_j)=u(c-1)\gamma$}, $1\leq i\neq j\leq \delta+1$.
\end{cor}
\begin{proof}
The first part follows  easily from the proof of the above theorem. Next,
since $\epsilon\in\{\alpha-c\beta\,: \, \alpha,\beta\in\F_{q},\alpha-c\beta\neq0\}$, there are at most $q^2-1$ such sets ($q-1$ if $c\in\F_q$), so $\delta+1\leq q^2-1$ ($\delta+1\leq q-1$ if $c\in\F_q$), so the $c$-differential uniformity of $H$ is at most $q^2-1$, and $q-1$ for $c\in\F_q$.

As for $F(x)=\sum_{i=0}^{n-1} a_i x^{q^i}$, 
since $F$ is PcN and all pairs $x_i,x_j$ ($i\neq j$) spread in different sets $A_{a,\epsilon}$, then by Dirichlet's box principle there must exist two elements $\alpha-c\beta,\alpha'-c\beta'$ with $\alpha-\alpha'=\beta-\beta'=\gamma$.
\end{proof}

 \section{Some classes of PcN and APcN via the switching method}
 
 We  give here some new infinite classes of APcN functions (most of them in even characteristic), with respect to all $c\neq 1$, constructed via our switching method. Recall that $\Trn$ is the relative trace function from $\F_{q^n}$ to $\F_{q}$. 
\begin{thm}
  \label{thm:sum_traces}
 Let $p$ prime, $\ell\geq 2$, $q=p^\ell$, $n=mp$, $m\geq 1$, $u\in\F_q^*$, $L$, a $q$-linearized permutation polynomial, and the  function  $H_0(x)=L(x)+ u \Trn(x)$ over $\F_{q^n}$. Let, for $q=2,k\geq 1$, $L$ be a  $q$-linearized permutation polynomial such that $L(-1)=-1$. Let $H_k(x)=L(x)+\Trn(x^{2^k+1})$, $H_{k_1<k_2<\ldots<k_s}(x)=L(x)+\prod_{i=1}^s\left(\Trn(x^{2^{k_i}+1}+\delta_i)\right)^{g_i}$, $g_i\in\N,\,\delta_i\in\F_{2^n}$, $1\leq k_i \leq n-1$,  on $\F_{2^n}$,  $n\geq 3$, and  $G_{k_1<k_2<\ldots<k_s}(x)=L(x)+\prod_{i=1}^s\left(\alpha_i\Tr_{q^n/q^m}(x^{2^{k_i}+1}+\delta_i)\right)^{g_i}$, $g_i\in\N,\,\delta_i\in\F_{2^n}$, $\alpha_i\in\F_{2^m}^*$, $1\leq k_i \leq n-1$, on $\F_{2^n}$,  $n\geq 3$. Then $H_0$ is PcN,  $H_k,H_{k_1<k_2<\ldots<k_s}$ and  $G_{k_1<k_2<\ldots<k_s}$ are  either APcN or PcN with respect to all $c\neq 1$, and PcN for $c=0$.
 \end{thm}
 \begin{proof}
 We first concentrate on $H_0$, and fix $a,b\in\F_{q^n}$.
 We shall use Theorem~\ref{thm:pcn}, with $u\in\F_q^*,\delta=1, F(x)=L(x), f(x)= \Trn(x)$, and any $c\neq1$. 
 If  ${_c}D_a F(x)-{_c}D_a F(y)=(1-c) L(x-y) =-u\epsilon$, for $\epsilon\in\{\alpha-c\beta\, : \, \alpha,\beta\in\F_{q},\alpha-c\beta\neq0\}$, and ${_c}D_a f(x)-{_c}D_a f(y)=\epsilon$, then ${_c}D_a f(x)-{_c}D_a f(y)=(1-c)\Trn(x-y)=(1-c) \Trn\left(L^{-1}\left(\frac{-u\epsilon}{1-c}\right)  \right)$. 
 If $c\in\F_q$, then 
 \[
 {_c}D_a f(x)-{_c}D_a f(y)=\Trn\left((1-c)L^{-1}(-u\epsilon)\right)=\epsilon.
 \]
 It follows that $\epsilon\in\F_q$, but then $\Trn\left(L^{-1}(-u\epsilon) \right)=L^{-1}(-u\epsilon)\Trn(1)=L^{-1}(-u\epsilon) n=0\neq \epsilon$. 
 If $c\not\in\F_q$, then $1,c$ are independent over $\F_q$, and so, $(1-c)\Trn\left(L^{-1}\left(\frac{-u\epsilon}{1-c}\right)  \right)=\epsilon=\alpha-c\beta$, would imply $\alpha=\beta=\Trn\left(L^{-1}\left(\frac{-u\epsilon}{1-c}\right)  \right)=\Trn(L^{-1}(-u\alpha))=L^{-1}(-u\alpha)\Trn(1)=L^{-1}(-u\alpha)n=0$, so $\epsilon=0$, a contradiction.
 
 Although, we can embed the case of $H_k$ in the general $H_{k_1<k_2<\ldots<k_s}$ (taking $s=1$), we prefer to give a detailed proof since it will reveal the method we use for the last claim.
 For arbitrary $1\leq k\leq n-1$,  we now use Theorem~\ref{thm:gen}, with $u=1,\delta=2, F(x)=L(x), f(x)= \Trn(x^{2^k+1})$ to deal with $H_k$.
Assuming that $H_k$ is not APcN and so, for fixed $a$, by Corollary \ref{cor}, there exist $x_i$, $1\leq i\leq 3$, all in different $A_{a,\epsilon}$  such that for  $1\leq i\neq j\leq 3$,
 \begin{align*}
 & {_c}D_a F(x_i)-{_c}D_aF(x_j)=(1-c) L(x_i-x_j)=-u\epsilon\\
 & {_c}D_a f(x_i)-{_c}D_af(x_j)=\epsilon\neq0.
 \end{align*}
Using again Corollary \ref{cor}, taking $\Gamma=1$, there must be a pair, say $x_1,x_2$, such that $(1-c) L(x_1-x_2)=-(1-c)$, and so, since $L$ is a permutation with $L(-1)=-1$, $x_1-x_2=-1$. Further, using the fact that $ \Trn(x^{2^k}a)= \Trn(x a^{2^{-k}})$, then
 \begin{align*}
 {_c}D_af(x_i)&=(1-c) \Trn(x_i^{2^k+1})\\
 &\qquad + \Trn(x_i^{2^k} a+a^{2^k}x_i+a^{2^k+1}), i=1,2\\
 {_c}D_af(x_1)-{_c}D_af(x_2)&=(1-c)  \Trn(x_1^{2^k+1}+x_2^{2^k+1})\\
 &\qquad + \Trn((x_1-x_2)(a^{2^{-k}}+a^{2^{k}}))\\
&=(1-c) \Trn(-x_1^{2^k}-x_1-1)+ \Trn(-a^{2^{-k}}-a^{2^{k}})\\
&=-(1-c) \Trn(x_1^{2^k}+x_1)-(1-c) \Trn(1)=0\neq\epsilon,
 \end{align*}
 and the second claim is shown.
 
 The last two claims follow in a similar fashion. For $H_{k_1<k_2<\ldots<k_s}$, we observe that, denoting $f(x)=\prod_{i=1}^s\left( \Trn(\alpha_ix^{p^{k_i}+1}+\delta_i)\right)^{g_i}$, then, as above, there is a pair of elements, say $x_1,x_2$ such that $x_1-x_2=-1$ for which ${_c}D_a F(x_1)-{_c}D_aF(x_2)=(1-c) L(x_1-x_2)=-(1-c)$. Now, for any $z$,
\begin{align*}
f(z-1)&=\prod_{i=1}^s\left( \Trn((z-1)^{p^{k_i}+1}+\delta_i)\right)^{g_i}\\
&=\prod_{i=1}^s\left( \Trn(z^{p^{k_i}+1}-z^{p^{k_j}}-z+1)+ \Trn(\delta_i)\right)^{g_i}\\
&=\prod_{i=1}^s\left( \Trn(z^{p^{k_i}+1}+\delta_i)\right)^{g_i}=f(z),
\end{align*}
 and so,
 \begin{align*}
 {_c}D_af(x_1)&-{_c}D_af(x_2)=f(x_1+a)-cf(x_1)-(f(x_2+a)-cf(x_2))\\
 &\qquad =f(x_1+a)-f(x_1+a-1)-c\left( f(x_1)-f(x_1+1)\right)=0,
 \end{align*}
 and thus our $H_{k_1<k_2<\ldots<k_s}$ must be either APcN or PcN.

The proof for $G_{k_1<k_2<\ldots<k_s}$ is similar to that of $H_{k_1<k_2<\ldots<k_s}$.
\end{proof}

In~\cite{LWWZ20}, for any characteristic $p$, seven polynomials classes of the form $x+\left(\Tr_{q^n/q^m}(x) ^k+\gamma\right)^s$, as well as 
three classes of the form $x+\left(\Tr_{q^n/q^m}(x) ^{k_1}+\gamma\right)^{s_1}+\left(\Tr_{q^n/q^m}(x) ^{k_2}+\gamma\right)^{s_2}$, where $m|n$, are shown to be permutation polynomials (mostly, for $n=2m$). With an eye towards these classes, we show the following result, which in some instances will reprove some results of~\cite{LWWZ20} (recall that PcN with respect to $c=0$ is simply the permutation property).
\begin{thm}
\label{thm:LW}
Let $p$ be a prime number, $m,n$ positive integers such that $m\,|\,n$ and $p\,|\,\frac{n}{m}$, $1\leq t\in\mathbb{Z}_{>0}$, $u\in\F_{p^m}^*$, $\delta_i\in\F_{p^m}$, 
$1\leq k_i,s_i\leq p^n-1$, $1\leq i\leq t$, and $L$ a linearized permutation polynomial on $\F_{q^n}$. Then, the functions $\displaystyle H(x)=L(x)+u\sum_{i=1}^t\left(\Tr_{q^n/q^m}(x)^{k_i}+\delta_i \right)^{s_i}$  are  PcN, with respect to all $c\in \F_{p^m}\setminus\{1\}$.
\end{thm}
\begin{proof}
As we have become accustomed, we use our switching Theorem~\ref{thm:pcn} with 
\[
u\in\F_{p^m}^*,F(x)=L(x), f(x)=\sum_{i=1}^t\left(\Tr_{q^n/q^m}(x)^{k_i}+\delta_i \right)^{s_i}.
\]
 If there are two elements $x_1,x_2$ satisfying those conditions of Theorem~\ref{thm:pcn}, then $(1-c) L(x_1-x_2) =-u(\alpha-c\beta)$, for some $\alpha,\beta\in\F_{p^m}$, not both zero, and so, $x_2=x_1+\gamma$, where $\gamma=L^{-1}\left(\frac{u(\alpha-c\beta)}{1-c}\right)\in \F_{p^m}^*$.

Furthermore,
\begin{align*}
f(x+\gamma)&=\sum_{i=1}^t\left(\Tr_{q^n/q^m}(x+\gamma)^{k_i}+\delta_i \right)^{s_i}\\
&=\sum_{i=1}^t\left(\left(\Tr_{q^n/q^m}(x)+\Tr_{q^n/q^m}(\gamma)\right)^{k_i}+\delta_i \right)^{s_i}=f(x),
\end{align*}
 since $n/m$ is divisible by $p$, and so, $\Tr_{q^n/q^m}(\gamma)=\gamma\Tr_{q^n/q^m}(1)=0$, for $\gamma\in\F_{p^m}$.
Therefore, ${_c}D_a f(x_1)-{_c}D_af(x_2)=f(x_1+a)-f(x_1+a+\gamma)-c\left(f(x_1)-f(x_1+\gamma)\right)=0$, and so, it cannot equal ${_c}D_a F(x_2)-{_c}D_aF(x_1)=(1-c)\gamma\neq 0$. 
In fact, the equation ${_c}D_aF(x_1)=b-{_c}D_a f(x_1)=b-{_c}D_a f(x_2)={_c}D_aF(x_2)$ holds, but that is impossible if $x_1\neq x_2$, since $F$ is PcN. The claim is shown.
\end{proof}

\begin{rem}
 If   $m=1$ \textup{(}hence, the parameters $s_i,k_i$ can be taken to be all~$1$\textup{)}, the previous proof can be easily rewritten, using Theorem~\textup{\ref{thm:gen}}, to show that the classes of Theorem~\textup{\ref{thm:LW}} are either APcN or PcN, for all $c\neq 1$ \textup{(}thus, complementing Theorem~\textup{\ref{thm:sum_traces}}\textup{)}.
\end{rem}

\section{Further comments}

In this paper we generalize Dillon's switching method (see Edel and Pott~\cite{EP09} for details on it) in even and odd characteristic, to provide necessary and sufficient conditions for a function modified in some components to have $c$-differential uniformity equal to $\delta$. We then use this approach to provide many classes of PcN and APcN for even and odd characteristic. As a by-product we generalize some prior results. It would surely  be interesting if more classes of functions can be found via this (or similar) component modifications of a vectorial Boolean or $p$-ary function.

\end{document}